\documentclass[runningheads]{llncs}
\pdfoutput=1

\usepackage[T1]{fontenc}
\usepackage{graphicx}
\usepackage{amssymb}
\usepackage{subcaption}
\usepackage{thm-restate}
\usepackage{amsmath}

\usepackage{booktabs}
\usepackage{multirow}
\usepackage{colortbl}
\usepackage{wrapfig}
\usepackage{listings}
\usepackage{ltl}

\definecolor{linkC}{HTML}{710580}

\usepackage{hyperref}
\hypersetup{
	colorlinks,
	citecolor=linkC,
	filecolor=red,
	linkcolor=linkC,
	urlcolor=blue
}
\usepackage{cleveref}
\usepackage{color}

\usepackage{pifont}

\usepackage{orcidlink}
\usepackage{marvosym}

\newcommand{\cmark}{\ding{51}}%
\newcommand{\xmark}{\ding{55}}%

\newcommand{\myrightarrow}[1]{\mathrel{\raisebox{-1pt}{$\xrightarrow{#1}$}}}

\newcommand{\LTL}{LTL}
\newcommand{\HyperLTL}{HyperLTL}

\newcommand{\ltln}{\LTLnext}
\newcommand{\ltlg}{\LTLsquare}

\newcommand{\ltlu}{\LTLuntil}

\newcommand{\nat}{\mathbb{N}}
\newcommand{\ldot}{\mathpunct{.}}
\newcommand{\lang}[1]{\mathcal{L}(#1)}
\newcommand{\ltlmodels}{\models}

\newcommand{\traceVars}{\mathcal{V}}
\newcommand{\traceSet}{\mathbb{T}}
\newcommand{\traces}[1]{\mathit{Traces}(#1)}
\newcommand{\ap}{\mathit{AP}}

\newcommand{\calT}{\mathcal{T}}

\newcommand{\calA}{\mathcal{A}}

\newcommand{\calL}{\mathcal{L}}

\newcommand{\zip}{\mathit{zip}}
\newcommand{\quant}{\mathbb{Q}}

\newcommand{\spot}{\texttt{spot}}
\newcommand{\bait}{\texttt{BAIT}}
\newcommand{\rabit}{\texttt{RABIT}}
\newcommand{\forklift}{\texttt{FORKLIFT}}

\newcommand{\AutoHyper}{\texttt{AutoHyper}}

\newif\iffullversion
\fullversiontrue

\newcommand{\ifFull}[2]{\iffullversion#1\else#2\fi}

\begin{document}

\title{AutoHyper: Explicit-State Model \\Checking for HyperLTL}
\titlerunning{AutoHyper: Explicit-State Model Checking for HyperLTL}

\author{}
\institute{}

\author{
Raven Beutner
\and
Bernd Finkbeiner
}

\institute{CISPA Helmholtz Center for Information Security, \\Saabrücken, Germany}

\author{Raven Beutner\textsuperscript{(\Letter)} \!\!\! \scalebox{1.25}{\orcidlink{0000-0001-6234-5651}}  \and
	Bernd Finkbeiner \!\!\! \scalebox{1.25}{\orcidlink{0000-0002-4280-8441}} }

\authorrunning{R.~Beutner and B.~Finkbeiner}

\institute{CISPA Helmholtz Center for Information Security, Germany\\ \email{\{raven.beutner,finkbeiner\}@cispa.de}}

\maketitle

\begin{abstract}
HyperLTL is a temporal logic that can express hyperproperties, i.e., properties that relate multiple execution traces of a system.
Such properties are becoming increasingly important and naturally occur, e.g., in information-flow control, robustness, mutation testing, path planning, and causality checking.
Thus far, complete model checking tools for HyperLTL have  been limited to alternation-free formulas, i.e., formulas that use only universal or only existential trace quantification. 
Properties involving quantifier alternations could only be handled in an incomplete way, i.e., the verification might fail even though the property holds. 
In this paper, we present \texttt{AutoHyper}, an explicit-state automata-based model checker that supports full HyperLTL and is complete for properties  with arbitrary quantifier alternations.
We show that language inclusion checks can be integrated into HyperLTL verification, which allows \texttt{AutoHyper} to benefit from a range of existing inclusion-checking tools.  
We evaluate \texttt{AutoHyper} on a broad set of benchmarks drawn from different areas in the literature and compare it with existing (incomplete) methods for HyperLTL verification.
\end{abstract}

\section{Introduction}

Hyperproperties \cite{ClarksonS08} are system properties that relate multiple executions of a system.
Such properties are of increasing importance as they naturally occur, e.g., in information-flow control \cite{Rabe16}, robustness \cite{DArgenioBBFH17}, linearizability \cite{HerlihyW90,HsuSB21}, path planning \cite{0044NP20}, mutation testing \cite{FellnerBW21}, and causality checking \cite{CoenenFFHMS22}.
A prominent logic to express hyperproperties is \HyperLTL{}, which extends linear-time temporal logic (\LTL) with explicit trace quantification \cite{ClarksonFKMRS14}.
\HyperLTL{} can, for instance, express generalized non-interference (GNI) \cite{McCullough88}, stating that the high-security input of a system does not influence the observable output.
\begin{align}\label{prop:GNIintro}
	\forall \pi\ldot \forall \pi'\ldot \exists \pi''\ldot \ltlg\Big(\bigwedge_{a \in H} a_{\pi} \leftrightarrow a_{\pi''}\Big) \land \ltlg\Big(\bigwedge_{a \in L \cup O} a_{\pi'} \leftrightarrow a_{\pi''}\Big)\tag{\text{GNI}}
\end{align}
Here, $H$ is a set of high-security input, $L$ is a set of low-security inputs, and $O$ is a set of low-security outputs.
The formula states that for any traces $\pi, \pi'$ there exists a third trace $\pi''$ that agrees with the high-security inputs of $\pi$ and with the low-security inputs and outputs of $\pi'$.
Any observation made by a low-security attacker is thus compatible with every possible high-security input.

We are interested in the model checking (MC) problem of HyperLTL, i.e., whether a given (finite-state) system satisfies a given property. 
For HyperLTL, the structure of the quantifier prefix directly impacts the complexity of this problem.
For alternation-free formulas (i.e., formulas that only use quantifiers of a single type), verification is well understood and is reducible to the verification of a trace property on a self-composition of the system \cite{BartheDR11}.
This reduction has, for example, been implemented in \texttt{MCHyper} \cite{FinkbeinerRS15}, a tool that can model check (alternation-free) HyperLTL formulas in systems of considerable size (circuits with thousands of latches).

Verification is much more challenging for properties involving quantifier alternations (such as \ref{prop:GNIintro} from above).
While MC algorithms supporting full HyperLTL exist (see \cite{ClarksonFKMRS14,FinkbeinerRS15}), they have not been implemented yet. 
Instead, over the years, a number of approaches to the verification of such properties in practice have been made:
Finkbeiner et al.~\cite{FinkbeinerRS15} and D'Argenio et al.~\cite{DArgenioBBFH17} manually strengthen properties with quantifier alternation into properties that are alternation-free and can be checked by \texttt{MCHyper}.
Coenen et al.~\cite{CoenenFST19} instantiate existential quantification in a $\forall^*\exists^*$ property (i.e., a property involving an arbitrary number of universal quantifiers followed by an arbitrary number of existential quantifiers, such as \ref{prop:GNIintro})  with an explicit (user-provided) strategy, thus reducing to the verification of an alternation-free formula.  
Alternatively, the strategy that resolves existential quantification can be automatically synthesized \cite{BeutnerF22b}.
Hsu et al.~\cite{HsuSB21} present a bounded model checking (BMC) approach for HyperLTL that is implemented in \texttt{HyperQube}. 
See \Cref{sec:relatedWork} for more details.

While all these verification tools can verify (or refute) interesting properties, they all suffer from the same fundamental limitation: they are \emph{incomplete}.
That is, for all the tools above, we can come up with verification instances where they fail, not because of resource constraints but because of inherent limitations in the underlying verification algorithm.
Moreover, such instances are not rare events but are encountered regularly in practice. 
For example, many of the benchmarks used to evaluate \texttt{HyperQube} (by Hsu et al.~\cite{HsuSB21}) do not admit a strategy to resolve existential quantification. Conversely, many of the properties verified by Coenen et al.~\cite{CoenenFST19} (such as \ref{prop:GNIintro}) cannot be verified using BMC \cite{HsuSB21}. 

\paragraph{AutoHyper.}
 
In this paper, we present \AutoHyper, a model checker for HyperLTL.
Our tool checks a hyperproperty by iteratively eliminating trace quantification using automata-complementations, thereby reducing verification to the emptiness check of an automaton \cite{FinkbeinerRS15}.
Importantly -- and different from previous tools for HyperLTL verification such as \texttt{MCHyper} \cite{FinkbeinerRS15,CoenenFST19} and \texttt{HyperQube} \cite{HsuSB21} -- \AutoHyper{} can cope with (and is \emph{complete} for) arbitrary HyperLTL formulas. 
Model checking using \AutoHyper{} does not require manual effort (such as writing an explicit strategy in \texttt{MCHyper} \cite{CoenenFST19}), nor does a user need to worry if the given property can even be verified with a given method.
\AutoHyper{} thus provides a ``push-button'' model checking experience for \HyperLTL{}.\footnote{The name of \AutoHyper{} is derived from the fact that it is both \textbf{Auto}mata-based and \textbf{Auto}matic (i.e., it is complete and does not require any user intervention). }

To improve \AutoHyper's efficiency, we make the (theoretical) observation that we can often avoid explicit automaton complementation and instead reduce to a language inclusion check on Büchi automata (cf. \Cref{prop:li}). 
On the practical side, this enables \AutoHyper{} to resort to a range of mature language inclusion checkers, including \spot{} \cite{Duret-LutzRCRAS22}, \rabit{} \cite{ClementeM19}, \bait{} \cite{DoveriGPR21}, and \forklift{} \cite{DoveriGM22}.

\paragraph{Evaluation.}

Using \AutoHyper{}, we extensively study the practical aspects of model checking HyperLTL properties with quantifier alternations. 
To evaluate the performance of explicit-state model checking, we apply \AutoHyper{} to a broad range of benchmarks taken from the literature and compare it with existing (incomplete) tools. 
We make the surprising observation that -- at least on the currently available benchmarks -- explicit-state MC as implemented in \AutoHyper{} performs on-par (and frequently outperforms) symbolic methods such as BMC \cite{HsuSB21}.
Our benchmarks stem from various areas within computer science, so \AutoHyper{} should -- thanks to its ``push-button'' functionality, completeness, and ease of use -- be a valuable addition to many areas.

Apart from using \AutoHyper{} as a practical MC tool, we can also use it as a complete baseline to systematically evaluate existing (incomplete) methods. 
For example, while it is known that replacing existential quantification with a strategy (as done by Coenen et al.~\cite{CoenenFST19}) is incomplete, it was, thus far, unknown if this incompleteness occurs frequently or is merely a rare phenomenon. 
We use \AutoHyper{} to obtain a ground truth and evaluate the strategy-based verification approach in terms of its effectiveness (i.e., how many instances it can verify despite being incomplete) and efficiency.

\paragraph{Structure.}

The remainder of this paper is structured as follows.
In \Cref{sec:prelims}, we introduce HyperLTL.
We recap automaton-based verification (which we abbreviate ABV) and our new approach utilizing language inclusion checks in \Cref{sec:recap}.
We discuss alternative verification approaches for HyperLTL in \Cref{sec:relatedWork}.
In \Cref{sec:eval1}, we compare different backend solving techniques and study the complexity of HyperLTL MC with multiple quantifier alternations  in practice; In \Cref{sec:eval2}, we evaluate ABV on a set of benchmarks from the literature and compare with the bounded model checker \texttt{HyperQube} \cite{HsuSB21}; In \Cref{sec:eval3} we use \AutoHyper{} for a detailed analysis of (and comparison with) strategy-based verification \cite{CoenenFST19,BeutnerF22b}.

\section{Preliminaries}\label{sec:prelims}

We fix a set of atomic propositions $\ap$ and define $\Sigma := 2^\ap$.
HyperLTL \cite{ClarksonFKMRS14} extends  LTL with explicit quantification over traces, thereby lifting it from a logic expressing trace properties to one expressing hyperproperties \cite{ClarksonS08}. 
Let $\traceVars$ be a set of trace variables. 
We define HyperLTL formulas by the following grammar:
\begin{align*}
 	\psi &:= a_\pi \mid \neg \psi \mid \psi \land \psi \mid \ltln \psi \mid \psi \ltlu \psi\\
 	\varphi &:= \exists \pi \ldot \varphi \mid \forall \pi \ldot \varphi \mid  \psi
\end{align*}%
where $\pi \in \traceVars$ and $a \in \ap$. 

We assume that the formula is closed, i.e., all trace variables that are used in the body are bound by some quantifier.
The semantics of HyperLTL is given with respect to a trace assignment $\Pi : \traceVars \rightharpoonup \Sigma^\omega$ mapping trace variables to traces. 
For $\pi \in \traceVars$ and $t \in \Sigma^\omega$, we write $\Pi[\pi \mapsto t]$ for the trace assignment obtained by updating the value of $\pi$ to $t$. 
Given a set of traces $\traceSet \subseteq \Sigma^\omega$, a trace assignment $\Pi$, and $i \in \nat$, we define:
\begin{align*}
	\Pi,i &\ltlmodels  a_\pi &\text{iff} \quad  &a \in \Pi(\pi)(i)\\
	\Pi,i &\ltlmodels  \neg \psi &\text{iff} \quad & \Pi,i \not\ltlmodels  \psi \\
	\Pi,i &\ltlmodels  \psi_1 \land \psi_2 &\text{iff} \quad  &\Pi,i \models \psi_1 \text{ and }  \Pi,i \models  \psi_2\\
	\Pi,i &\ltlmodels  \ltln  \psi &\text{iff} \quad & \Pi, i+1 \ltlmodels \psi \\
	\Pi,i &\ltlmodels  \psi_1 \ltlu \psi_2 &\text{iff} \quad & \exists j \geq i \ldot \Pi, j\ltlmodels  \psi_2 \text{ and } \forall i \leq k < j \ldot  \Pi, k \ltlmodels  \psi_1\\[3mm]
	\Pi &\models_\traceSet \psi  &\text{iff} \quad &\Pi, 0 \models \psi \\
	\Pi &\models_\traceSet \exists \pi \ldot \varphi &\text{iff} \quad &\exists t \in \traceSet \ldot \Pi[\pi \mapsto t] \models_\traceSet  \varphi\\
	\Pi &\models_\traceSet  \forall \pi \ldot \varphi &\text{iff} \quad &\forall t \in \traceSet \ldot \Pi[\pi \mapsto t] \models_\traceSet  \varphi
\end{align*}%

A \emph{transition system} is a tuple $\calT = (S, S_0, \kappa, L)$ where $S$ is a set of states, $S_0 \subseteq S$ is a set of initial states, $\kappa \subseteq S \times S$ is a transition relation, and $L : S \to \Sigma$ is a labeling function. 
We write $s \myrightarrow{\calT} s'$ whenever $(s, s') \in \kappa$.
A path is an infinite sequence $s_0s_1s_2 \cdots \in S^\omega$, s.t., $s_0 \in S_0$, and $s_i \myrightarrow{\calT} s_{i+1}$ for all $i$.
The associated trace is given by $L(s_0)L(s_1)L(s_2) \cdots \in \Sigma^\omega$. 
We write $\traces{\calT} \subseteq \Sigma^\omega$ for the set of all traces generated by $\calT$.
We say $\calT$ satisfies a HyperLTL property $\varphi$, written $\calT \models \varphi$, if $\emptyset \models_{\traces{\calT}} \varphi$, where $\emptyset$ denotes the empty trace assignment.

\section{Automata-based HyperLTL Model Checking}\label{sec:recap}

Given a system $\calT$ and HyperLTL property $\varphi$, we want to decide whether $\calT \models \varphi$.
In this section, we recap the automaton-based approach to the model checking of HyperLTL \cite{FinkbeinerRS15}. 
We further show how language inclusion checks can be incorporated into the model checking procedure to make use of a broad collection of mature language inclusion checkers.

\subsection{Automata-based Verification}\label{sec:automatonBasedVerification}

The idea of automata-based verification (ABV) \cite{FinkbeinerRS15} is to iteratively eliminate quantifiers and thus reduce MC to the emptiness check on an automaton. 
A non-deterministic Büchi automaton (NBA) is a tuple $\calA = (Q, Q_0, \delta, F)$ where $Q$ is a finite set of states, $Q_0 \subseteq Q$ is a set of initial states, $\delta : Q \times \Sigma \to 2^Q$ is a transition function, and $F \subseteq Q$ is a set of accepting states. 
We write $\calL(\calA) \subseteq \Sigma^\omega$ for the language of $\calA$, i.e., all infinite words that have a run that visits states in $F$ infinitely many times (see, e.g., \cite{baier2008principles}).
For traces $t_1, \ldots, t_n \in \Sigma^\omega$, we write $\zip(t_1, \ldots, t_n) \in (\Sigma^n)^\omega$ as the pointwise product, i.e., $\zip(t_1, \ldots, t_n)(i) := (t_1(i), \ldots, t_n(i))$.

Let $\calT = (S, S_0, \kappa, L)$ be a fixed transition system and let $\dot{\varphi}$ be some fixed closed HyperLTL formula (we use the dot to refer to the original formula and use $\varphi, \varphi'$ to refer to subformulas of $\dot{\varphi}$). 
For some subformula $\varphi$ that contains free trace variables $\pi_1, \ldots, \pi_n$, we say an NBA $\calA$ over $\Sigma^n$ is \emph{$\calT$-equivalent} to $\varphi$, if for all traces $t_1, \ldots, t_n$ it holds that $[\pi_1 \mapsto t_1, \ldots, \pi_n \mapsto t_n] \models_{\traces{\calT}} \varphi$ iff $\zip(t_1, \ldots, t_n) \in \calL(\calA)$.
That is, $\calA$ accepts exactly the zippings of traces that constitute a satisfying trace assignment for $\varphi$.

To check if $\calT \models \dot{\varphi}$, we inductively construct an automation $\calA_\varphi$ that is $\calT$-equivalent to $\varphi$ for each subformula $\varphi$ of $\dot{\varphi}$.
For the (quantifier-free) LTL body of $\dot{\varphi}$, we can construct this automaton via a standard LTL-to-NBA construction \cite{FinkbeinerRS15,baier2008principles}.
Now consider some subformula $\varphi' = \exists \pi. \varphi$ where $\varphi'$ contains free trace variables $\pi_1, \ldots, \pi_n$ and so $\varphi$ contains free trace variables $\pi_1, \ldots, \pi_n, \pi$.
We are given an inductively constructed NBA $\calA_{\varphi} = (Q, Q_0, \delta, F)$ over $\Sigma^{n+1}$ that is $\calT$-equivalent to $\varphi$.
We define the automaton $\calA_{\varphi'} $ over $\Sigma^n$ as $\calA_{\varphi'} := (S \times Q, S_0 \times Q_0, \delta', S \times F)$ where $\delta'$ is defined as
\begin{align*}
	\delta'\Big((s, q), \big\langle l_1, \ldots, l_n\big\rangle \Big) := \Big\{(s', q') \mid s \xrightarrow{\calT} s' \;\land\; q' \in \delta\big(q, \big\langle l_1, \ldots, l_n, L(s)\big\rangle\big)	\Big\}.
\end{align*}
Informally, $\calA_{\varphi'}$ reads the zippings of traces $t_1, \ldots, t_n$ and guesses a trace $t \in \traces{\calT}$ such that $\zip(t_1, \ldots, t_n, t) \in \calL(\calA_{\varphi})$.
It is easy to see that $\calA_{\varphi'}$ is $\calT$-equivalent to $\varphi'$.
To handle universal trace quantification, we consider a formula $\varphi' = \forall \pi. \varphi$ as $\varphi' = \neg \exists \pi. \neg \varphi$ and combine the construction for existential quantification with an automaton complementation.

Following the inductive construction, we obtain an automaton $\calA_{\dot{\varphi}}$ over the singleton alphabet $\Sigma^0$ that is $\calT$-equivalent to $\dot{\varphi}$.
By definition of $\calT$-equivalence,  $\calT \models {\dot{\varphi}}$ iff $\emptyset \models_{\traces{\calT}} \dot{\varphi}$ iff $\calA_{\dot{\varphi}}$ is non-empty (which we can decide \cite{CourcoubetisVWY92}). 

\subsection{HyperLTL Model Checking by Language Inclusion}\label{sec:languageInclusion}

The algorithm outlined above requires one complementation for each quantifier alternation in the HyperLTL formula. 
While we cannot avoid the theoretical cost of this complementation (see \cite{Rabe16,ClarksonFKMRS14}), we can reduce to a, in practice, more tamable problem: \emph{language inclusion}. 

For a system $\calT$, and a natural number $n \in \nat$ we define $\calA_\calT^n$ as an NBA over $\Sigma^n$ such that for any traces $t_1, \ldots, t_n \in \Sigma^\omega$ we have $\zip(t_1, \ldots, t_n) \in \lang{\calA_\calT^n}$ if and only if $t_i \in \traces{\calT}$ for every $1 \leq i \leq n$. 
We can construct $\calA_\calT^n$ by building the $n$-fold self-composition of $\calT$ \cite{BartheDR11} and convert this to an automaton by moving the labels from states to edges and marking all states as accepting.
We can now state a formal connection between language inclusion and HyperLTL MC (a proof can be found in \ifFull{Appendix \ref{app:langInc}}{the full version \cite{fullVersion}}):

\begin{restatable}{proposition}{langInc}\label{prop:li}
	Let $\dot{\varphi} = \forall \pi_1. \ldots \forall \pi_n. \varphi$ be a HyperLTL formula (where $\varphi$ may contain additional trace quantifiers) and let $\calA_\varphi$ be an automaton over $\Sigma^n$ that is $\calT$-equivalent to $\varphi$.
	Then $\calT \models \dot{\varphi}$ if and only if $\lang{\calA_\calT^n} \subseteq \lang{\calA_\varphi}$.
\end{restatable}

We can use \Cref{prop:li} to avoid a complementation for the outermost quantifier alternation. 
For example, assume $\dot{\varphi} = \forall \pi_1. \forall \pi_2. \exists \pi_3. \psi$ where $\psi$ is quantifier-free.
Using the construction from \Cref{sec:automatonBasedVerification}, we obtain an automaton $\calA_{\exists \pi_3. \psi}$ that is $\calT$-equivalent to $\exists \pi_3. \psi$ (we can construct $\calA_{\exists \pi_3. \psi}$ in linear time in the size of $\calT$).
By \Cref{prop:li}, we then have $\calT \models \dot{\varphi}$ iff $\lang{\calA^2_\calT} \subseteq \lang{\calA_{\exists \pi_3. \psi}}$.

Note that complementation and subsequent emptiness check is a theoretically optimal method to solve the (\texttt{PSPACE}-complete) language inclusion problem. 
\Cref{prop:li} thus offers no asymptotic advantages over ``standard'' ABV in \Cref{sec:automatonBasedVerification}.
In \emph{practice} constructing an explicit complemented automaton is often unnecessary as the language inclusion or non-inclusion might be witnessed without a complete complementation \cite{Duret-LutzRCRAS22,DoveriGPR21,ClementeM19,DoveriGM22}.
This makes \Cref{prop:li} relevant for the present work and the performance of \AutoHyper.

\section{Related Work and HyperLTL Verification Approaches}\label{sec:relatedWork}

HyperLTL \cite{ClarksonFKMRS14} is the most studied logic for expressing hyperproperties. 
A range of problems from different areas in computer science can be expressed as HyperLTL MC problems, including (optimal) path panning \cite{0044NP20}, mutation testing \cite{FellnerBW21}, linearizability \cite{HsuSB21}, robustness \cite{DArgenioBBFH17}, information-flow control \cite{Rabe16}, and causality checking \cite{CoenenFFHMS22}, to name only a few.
Consequently, any model checking tool for HyperLTL is applicable to many disciples within computer science and provides a unified solution to many challenging algorithmic problems. 
In recent years, different (mostly incomplete) methods for the verification of HyperLTL have been developed. 
We discuss them below (see \ifFull{Appendix \ref{app:overviewMC}}{the full version \cite{fullVersion}} for details).

\paragraph{Automata-based Model Checking.}

Finkbeiner et al.~\cite{FinkbeinerRS15} introduce the automata-based model checking approach as presented in \Cref{sec:automatonBasedVerification}.
For alternation-free formulas, the algorithms corresponds to the construction of the self-composition of a system \cite{BartheDR11} and is implemented in the \texttt{MCHyper} tool \cite{FinkbeinerRS15}.
\texttt{MCHyper} can handle systems of significant size (well beyond the reach of explicit-state methods) but is unable to handle any quantifier alternation (the main motivation for \AutoHyper).
\texttt{htltl2mc} \cite{ClarksonFKMRS14} is a prototype model checker for HyperLTL$_2$ (a fragment of HyperLTL with at most one alternation) built on top of \texttt{GOAL} \cite{TsaiTH13}.
In contrast to \texttt{htltl2mc}, \AutoHyper{} supports properties with arbitrarily many quantifier alternations and features automata with symbolic alphabets -- which is important to handle large systems with many atomic propositions, cf.~\Cref{fn:symbolic_alphabet}).

\paragraph{Strategy-based Verification.}

Coenen et al.~\cite{CoenenFST19} verify $\forall^*\exists^*$ properties by instantiating existential quantification with an explicit strategy.
This method -- which we refer to as strategy-based verification (SBV) -- comes in two flavors: 
either the strategy is provided by the user or the strategy is synthesized automatically. 
In the former case, model checking reduces to checking an alternation-free formula and can thus handle large systems, but requires significant user effort (and is thus no ``push-button'' technique).
In the latter case, the method works fully automatically \cite{BeutnerF22a,BeutnerF22b} but requires an expensive strategy synthesis.
SBV is incomplete as the strategy resolving existentially quantified traces only observes finite prefixes of the universally quantified traces.
While SBV can be made complete by adding prophecy variables \cite{BeutnerF22b}, the automatic synthesis of such prophecies is currently limited to very small systems and properties that are temporally safe \cite{BeutnerCFHK22}.
We investigate both the performance and incompleteness of SBV in \Cref{sec:eval3}.

\paragraph{Bounded Model Checking.}

Hsu et al.~\cite{HsuSB21} propose a bounded model checking (BMC) procedure for HyperLTL.
Similar to BMC for trace properties \cite{BiereCCZ99}, the system is unfolded up to a fixed depth, and pending obligations beyond that depth are either treated pessimistically (to show the satisfaction of a formula) or optimistically (to show the violation of a formula).
While BMC for trace properties reduces to SAT-solving, BMC for hyperproperties naturally reduces to QBF-solving. 
As usual for bounded methods, BMC for HyperLTL is incomplete. 
For example, it can never show that a system satisfies a hyperproperty where the LTL body contains an invariant (as, e.g., is the case for \ref{prop:GNIintro}).\footnote{BMC for trace properties can be made complete by using bounds on the unrolling depth (also called completeness thresholds) \cite{ClarkeKOS04} and including loop conditions in the encoding \cite{BiereCCZ99}. As remarked by Hsu et al.~\cite{HsuSB21}, the same is much more challenging for hyperproperties, and no solutions have been proposed.
Instead, Hsu et al.~\cite{HsuSB21} propose an alternative unrolling semantics (which they call halting semantics) that can mitigate this incompleteness issue for programs that terminate after a \emph{fixed} number of steps. This is a strong (and often unrealistic) assumption for general reactive systems. }
We compare \AutoHyper{} and BMC (in the form of \texttt{HyperQube} \cite{HsuSB21}) in \Cref{sec:eval2}.

\section{AutoHyper: Tool Overview}

\AutoHyper{} is written in \texttt{F\#} and implements the automaton-based verification approach described in \Cref{sec:automatonBasedVerification} and, if desired by the user, makes use of the language-inclusion-based reduction from \Cref{sec:languageInclusion}.
\AutoHyper{} uses \spot{} \cite{Duret-LutzRCRAS22} for LTL-to-NBA translations and automata complementations.
To check language inclusion, \AutoHyper{} uses \spot{} (which is based on determinization), \rabit{} \cite{ClementeM19} (which is based on a Ramsey-based approach with heavy use of simulations), \bait{} \cite{DoveriGPR21}, and \forklift{} \cite{DoveriGM22} (both based on well-quasiorders).
\AutoHyper{} is designed such that communication with external automata tools is done via established text-based formats (opposed to proprietary APIs), namely the \texttt{HANOI} \cite{BabiakBDKKM0S15} and \texttt{BA} automaton formats. New (or updated) tools that improve on fundamental automata operations, such as complementation and inclusion checks, can thus be  integrated easily. 
Internally we represent automata using symbolic alphabets (similar to \spot{}).
We store transition formulas as DNFs as this allows for very efficient SAT checks, which are needed during the product construction. 

All experiments in this paper were conducted on a Mac Mini with an Intel Core i3 (i3-8100B) and 16GB of memory.
We used \spot{} version 2.11.1; \rabit{} version 2.4.5; \bait{} commit 369e1a4; and \forklift{} commit 5d519e3.

\paragraph{Input Formats.}

\AutoHyper{} supports both explicit-state systems (given in a \texttt{HANOI}-like \cite{BabiakBDKKM0S15} input format) and symbolic systems that are internally converted to an explicit-state representation. 
The support for symbolic systems includes Aiger circuits, symbolic models written in a fragment of the NuSMV input language \cite{nusmv}, and a simple boolean programming language \cite{BeutnerF21}.

\paragraph{Random Benchmarks.}

For our evaluation, we use both existing instances from various sources in the literature and randomly generated problems.\footnote{The advantage of randomly generated instances is twofold.
	First, it allows for the easy generation of a large set of benchmarks.
	Second, the random generation is parameterized by multiple parameters (such as system size, transition density, formula size, etc.), enabling a comprehensive analysis of the exact impact of different parameters on the model checking complexity in practice. }
We generate random transition systems based on the Erdős–Rényi–Gilbert model \cite{fienberg2012brief}.
Given a size $n$ and a density parameter $p \in [0, 1]$, we generate a graph with $n$ states, where for every two states $s, s'$, there is a transition $s \to s'$ with probability $p$.
To generate a graph with $n$ edges and, in expectation, constant outdegree of $k$, we can choose $p = \tfrac{k}{n}$.
We further ensure that the system is connected and all states have at least one outgoing edge.
We generate random HyperLTL formulas (with a given quantifier prefix) by sampling the LTL matrix using \spot{}'s \texttt{randltl}.

\section{HyperLTL Model Checking Complexity in Practice}\label{sec:eval1}

Before we turn our attention to benchmarks found in the literature, we compare the different backend inclusion checkers supported by \AutoHyper{} by evaluating them on a large set of synthetic (random) benchmarks (in \Cref{sec:backendComp}).
Moreover, the random generation of benchmarks allows us to peek at formulas with more than one quantifier alternation.
The theoretical hardness of model checking properties with multiple alternations has been studied extensively \cite{ClarksonFKMRS14,Rabe16}, and we analyze, for the first time, how these results transfer to practice (in \Cref{sec:qunatifierAlternation}). 

\subsection{Performance of Inclusion Checkers }\label{sec:backendComp}

As the first set of benchmarks, we compare the different backend inclusion checkers supported by \AutoHyper.
In \Cref{fig:largerInstances}, we depict how many instances can be solved using the inclusion checks of \spot{}, \bait{}, \rabit{}, and \forklift{} within a timeout of 10s and give the median running time used on the instances that could be solved within the timeout.
We observe that \spot{} clearly outperforms \rabit{}, \bait{}, and \forklift{} in terms of the percentage of instances that can be checked within 10s.\footnote{We remark that \spot{} operates on automata with a symbolic alphabet (i.e., transitions are defined as boolean formulas over $\ap$). In contrast, \rabit{}, \bait{}, and \forklift{} only support explicit alphabets (i.e., automata with one symbol for each element in $2^\ap$).
}
While, in general, \spot{} solves the most instances, a manual inspection reveals that there are also instances that can only be solved by \rabit{} or \bait{}/\forklift{}.
This justifies why \AutoHyper{} supports multiple backed inclusion checkers that implement different algorithms and thus excel on different problems (we will confirm this in \Cref{sec:eval2}). 
Moreover, our experiments provide evidence that HyperLTL MC is a natural source for challenging language inclusion benchmarks (see \ifFull{Appendix \ref{app:liBench}}{the full version \cite{fullVersion}}).

\begin{figure}[!t]
	\centering
	\vspace{-3mm}
	\includegraphics[width=1\linewidth]{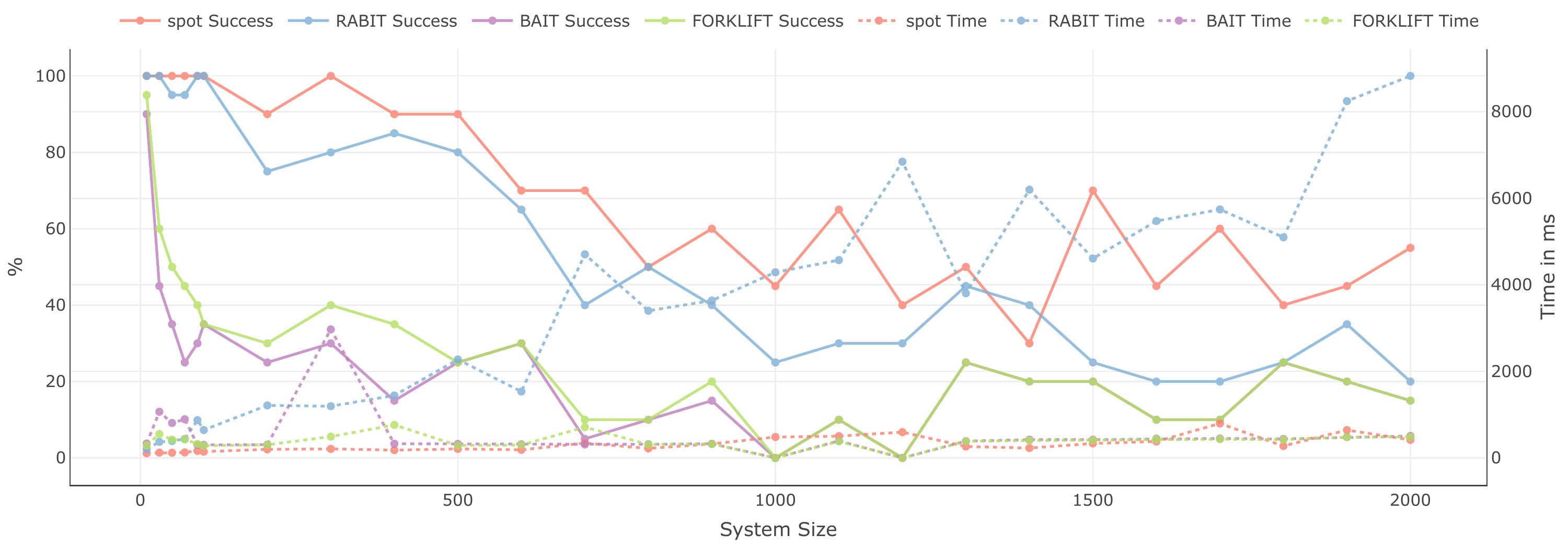}
	\vspace{-5mm}
	\caption{We evaluate different backend solvers on instances of varying system size with an (on average) constant outdegree of 10 and a fixed property size of 20.
		We generate 20 samples per system size. 
		We display both the success rate of each solver within a timeout of 10s (on the left axis) and the median running time on the solved instances (on the right axis). 
		\vspace{-2mm}
	}\label{fig:largerInstances}
\end{figure}

\begin{wrapfigure}{R}{0.45\textwidth}
	\vspace{-10mm}
	\begin{center}
		\includegraphics[width=0.9\linewidth]{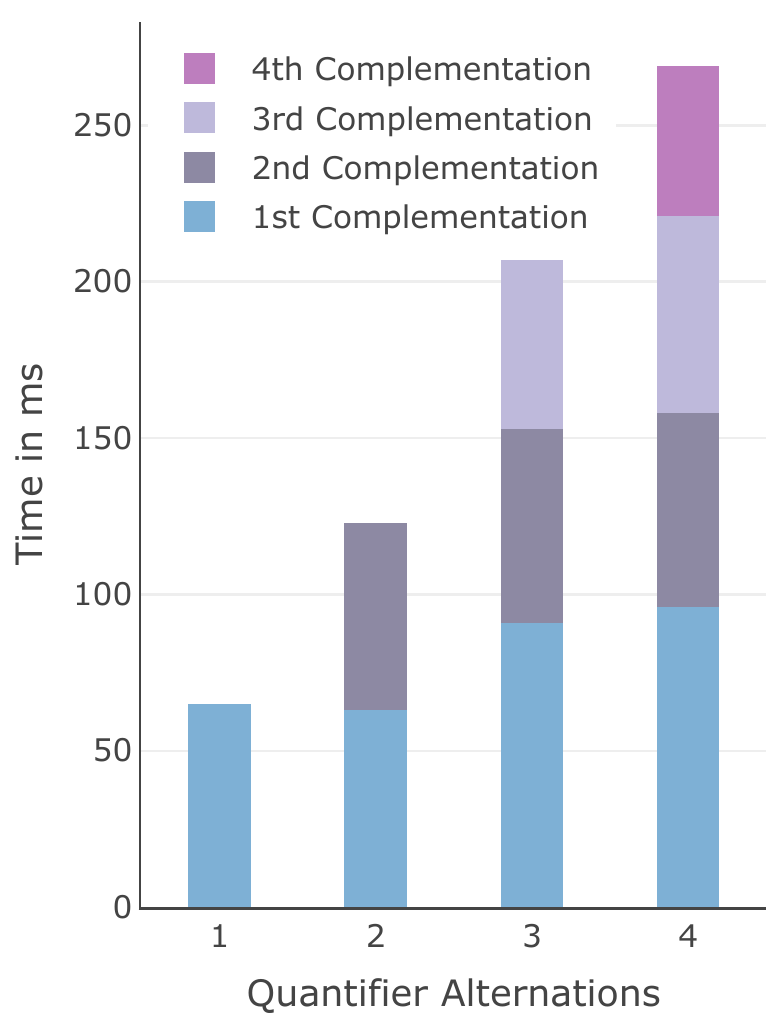}
	\end{center}	
	\vspace{-5mm}
	\caption{For properties with a varying number of quantifier alterations, we display the average time spent on the automata complementation during model checking. 
		\vspace{-16mm} }\label{fig:alternationCompCost}
\end{wrapfigure}

We remark that we set the timeout of 10s deliberately low to compute (and reproduce) the plots in a reasonable time (computing \Cref{fig:largerInstances} took about 3.5h).
If a user wants to verify a given instance and does not require a result within a few seconds, running the solver for even longer will likely increase the success rate further (see also the evaluation in \Cref{sec:eval2}).

\subsection{Model Checking Beyond $\forall^*\exists^*$}\label{sec:qunatifierAlternation}

Using randomly generated benchmarks, we can also peek at the practical complexity of model checking in the presence of multiple quantifier alternations.
In \emph{theory}, the model checking complexity of HyperLTL increases by one exponent with each quantifier alternation \cite{ClarksonFKMRS14,Rabe16}. 
Using \AutoHyper{}, we can, for the first time, investigate the model checking complexity in \emph{practice}. 

We model check randomly generated formulas with $1$ to $4$ quantifier alternations and visualize the total running time based on the cost of each complementation (using \spot{}) in \Cref{fig:alternationCompCost} (recall that checking a formula with $k$ alternations using ABV requires $k$ automaton complementations).
Although the number of quantifier alternations has an undeniable impact on the total running time (the cumulative height of each bar), the increase in runtime is not proportional to the (non-elementary) increase suggested by the theoretical analysis. 
Different from the theoretical analysis (where the $(k+1)$th complementation is exponentially more expensive than the $k$th), the cost of each complementation barely increases (or even decreases).
This suggests that the $\calT$-equivalent automata constructed in each iteration are, in practice, much smaller than indicated by the worst-case theoretical analysis.
Verification of properties beyond one alternation is thus less infeasible than the theory suggests (at least on randomly generated test cases).

\section{Evaluation on Symbolic Systems}\label{sec:eval2}

In this section, we challenge \AutoHyper{} with complex model checking problems found in the literature. 
Our benchmarks stem from a range of sources, including non-interference in boolean programs \cite{BeutnerF21}, symmetry in mutual exclusion algorithms \cite{CoenenFST19}, non-interference in multi-threaded programs \cite{SmithV98}, fairness  in non-repudiation protocols \cite{JamrogaMM11}, mutation testing \cite{FellnerBW21}, and path planning \cite{0044NP20}.

\subsection{Model Checking GNI on Boolean Programs}\label{sec:GNIInstances}

\begin{table}[!t]
	\caption{We depict the running time of \AutoHyper{} when verifying \ref{prop:GNIintro} on the boolean programs taken from \cite{BeutnerF21} and \cite{abs-2203-07283}. 
		We give the program, the bitwidth (bw), the size of the intermediate explicit-state representation (Size), and the time taken by each solver. 
		The timeout is set to 60s and indicated by a ``-''. The property holds in all cases. Times are given in seconds. }\label{tab:gniInstances}
	
	\begin{center}
		\begin{minipage}{0.5\textwidth}
			\centering
			\scalebox{0.83}{
				\def\arraystretch{1.2}
				\setlength\tabcolsep{3mm}
				\begin{tabular}{@{}c@{\hspace{2mm}}c@{\hspace{2mm}}c@{\hspace{2mm}}c@{\hspace{2mm}}c@{\hspace{2mm}}c@{\hspace{2mm}}c@{}}
					\toprule
					\textbf{Program} & \textbf{bw}  & \textbf{Size} & $\boldsymbol{t}_\texttt{spot}$ & $\boldsymbol{t}_\texttt{RABIT}$ & $\boldsymbol{t}_\texttt{BAIT}$ & $\boldsymbol{t}_\texttt{FORKLIFT}$ \\
					\midrule
					
					\multirow{3}{*}{\cite{BeutnerF21}.1}   & $1$-bit  & 17  & \textbf{0.52} & 0.59 & 0.80 & 0.61  \\
					&  $3$-bit &  65  & \textbf{0.56} & 0.86 & - & 22.73     \\
					&  $4$-bit &  129  & \textbf{0.99} & 5.51  & - & -    \\
					\arrayrulecolor{black!20}\midrule

					\multirow{1}{*}{\cite{BeutnerF21}.2}& $1$-bit  & 55 & \textbf{0.53}  & 0.69   & - & 5.49     \\
					\arrayrulecolor{black!20}\midrule

					\multirow{2}{*}{\cite{BeutnerF21}.3}  & $1$-bit  & 20  & \textbf{0.52}  & 0.61 & 3.05 & 0.98      \\
					&  $3$-bit &  80  & \textbf{0.61}  & 1.31 & - & -     \\
					\arrayrulecolor{black!20}\midrule

					\multirow{2}{*}{\cite{BeutnerF21}.4}  & $1$-bit  &  29 &  \textbf{0.52} & 0.56  & 0.58 & 0.57       \\
					&  $3$-bit &  113  & \textbf{0.67}  & 1.74 & - & -    \\
					\arrayrulecolor{black}\bottomrule
				\end{tabular}
			}
		\end{minipage}%
		\begin{minipage}{0.5\textwidth}
			\centering
			\scalebox{0.83}{
				\def\arraystretch{1.2}
				\setlength\tabcolsep{3mm}
				\begin{tabular}{@{}c@{\hspace{2mm}}c@{\hspace{2mm}}c@{\hspace{2mm}}c@{\hspace{2mm}}c@{\hspace{2mm}}c@{\hspace{2mm}}c@{}}
					\toprule
					\textbf{Program}  & \textbf{bw}  & \textbf{Size} & $\boldsymbol{t}_\texttt{spot}$ & $\boldsymbol{t}_\texttt{RABIT}$  & $\boldsymbol{t}_\texttt{BAIT}$  & $\boldsymbol{t}_\texttt{FORKLIFT}$  \\
					\midrule
					\multirow{1}{*}{\cite{abs-2203-07283}.1}  & $1$-bit  & 5 & 0.52  & \textbf{0.56}   & 0.58 & 0.57    \\
					\arrayrulecolor{black!20}\midrule
					
					\multirow{3}{*}{\cite{abs-2203-07283}.2} & $1$-bit  & 11 & \textbf{0.51}  & 0.57   & 0.72 & 0.61 \\
					& $2$-bit &  27  & \textbf{0.52} & 0.65  & 35.7 & 5.43  \\
					& $4$-bit &  291  & \textbf{1.46} & - & - & -   \\
					\arrayrulecolor{black!20}\midrule
					
					\multirow{2}{*}{\cite{abs-2203-07283}.3} & $1$-bit  & 21 & \textbf{0.52}  & 0.60 & 3.15 & 1.00   \\
					& $3$-bit &  225  & - & \textbf{45.2} & - & -  \\
					\arrayrulecolor{black!20}\midrule
					
					\multirow{2}{*}{\cite{abs-2203-07283}.4} & $1$-bit  & 25 & \textbf{0.52}  & 0.71    & 12.8 & 1.63  \\
					& $3$-bit &  193 & \textbf{0.98} & - & - & - \\
					
					\arrayrulecolor{black}\bottomrule
				\end{tabular}
			}
		\end{minipage}
	\end{center}
\end{table}

We use \AutoHyper{} to verify \ref{prop:GNIintro} on a range of boolean programs that process high-security and low-security inputs (taken from \cite{BeutnerF21,abs-2203-07283}).
\Cref{tab:gniInstances} depicts the runtime results using different backend solvers. 
We test each program with varying bitwidth and depict the largest bitwidth that can be solved by at least one solver (within a timeout of 60s).
We, again, note that \spot{} performs better than other inclusion checkers and, in particular, scales better when the size of the system increases.
Note that the number of atomic propositions is $3$ in all instances, so \spot{}'s support for symbolic alphabets has a negligible impact on the running time. 
We emphasize that not all instances in \Cref{tab:gniInstances} can be verified using SBV \cite{CoenenFST19,BeutnerF22b} without a user-provided fixed lookahead.
Likewise, BMC \cite{HsuSB21} can \emph{never} verify \ref{prop:GNIintro}.
This provides further evidence why complete model checking tools (of which \AutoHyper{} is the first) are necessary.

\subsection{Explicit Model Checking of Symbolic Systems}

\begin{table}[!t]
	\caption{We evaluate \texttt{HyperQube} and \AutoHyper{} on the benchmarks from \cite{HsuSB21}. We list the system and the property (as given in \cite[Table 2]{HsuSB21}), the quantifier structure ($\boldsymbol{Q}^*$), the verification result (Res) (\cmark{} indicates that the property holds and \xmark{} that it is violated), and the total running time of either tool ($t$).
		For \texttt{HyperQube}, we additionally list the unrolling bound ($k$) and the unrolling semantics (Sem).
		For \AutoHyper{}, we additionally list the size of the intermediate explicit state space (Size).
		Times are given in seconds. 
	}\label{tab:bmc}
	
	\begin{center}
		\small
		\def\arraystretch{1.2}
		\setlength\tabcolsep{3mm}
		\setlength\arrayrulewidth{0.7pt}
		
		\begin{tabular}{@{\hspace{3mm}}l@{\hspace{2mm}}l@{\hspace{4mm}}l@{\hspace{4mm}}l@{\hspace{6mm}}l@{\hspace{5mm}}l@{\hspace{5mm}}l@{\hspace{8mm}}l@{\hspace{5mm}}l@{\hspace{3mm}}}
			\toprule
			&& && \multicolumn{3}{@{}c@{\hspace{8mm}}}{\textbf{HyperQube} \cite{HsuSB21}} &\multicolumn{2}{c}{\textbf{AutoHyper}}\\
			\cmidrule[1pt](l{-2mm}r{6mm}){5-7}
			\cmidrule[1pt](l{-2mm}){8-9}
			\textbf{System} & \textbf{Spec} & $\boldsymbol{Q}^*$ & \textbf{Res} & $\boldsymbol{k}$ & \textbf{Sem} \!\!\!\!\! & $\boldsymbol{t}$ & \textbf{Size} & $\boldsymbol{t}$ \\
			\midrule
			Bakery$_3$ & $\varphi_{S1}$ & $\exists\exists$ & \xmark & 7 & pes  & \textbf{1.9} & 167  & 2.3 \\
			Bakery$_3$ & $\varphi_{S2}$ & $\forall\exists$ & \xmark & 12 & pes & \textbf{2.0} & 167  & 4.2 \\
			Bakery$_3$ & $\varphi_{S3}$ & $\exists\forall$ & \xmark$^!$ & 20 & pes  & \textbf{2.8} & 167  & 34.6 \\
			Bakery$_3$ & $\varphi_{\mathit{sym}1}$ \!\!\!\!\!&$\forall\exists$ & \xmark & 10 & pes  & \textbf{1.7} & 167  & 16.2\\
			Bakery$_3$ & $\varphi_{\mathit{sym}2}$\!\!\!\!\! & $\forall\exists$& \xmark & 10 & pes  & \textbf{1.6} & 167  & 2.9 \\
			Bakery$_5$ & $\varphi_{\mathit{sym}1}$ \!\!\!\!\!&$\forall\exists$ & \xmark & 10  & pes  & \textbf{17.3} & 996  & 282.1 \\
			Bakery$_5$ & $\varphi_{\mathit{sym}2}$\!\!\!\!\! &$\forall\exists$ & \xmark & 10 & pes  & 18.2 & 996   & \textbf{18.0}  \\
			\arrayrulecolor{black!20}\midrule
			SNARK-bug1   & $\varphi_{\mathit{lin}}$ &$\forall\exists$ & \xmark & 26  & hpes  & 618.0 & 4941 & \textbf{96.1} \\
			\arrayrulecolor{black!20}\midrule
			3-Thread$_\mathit{correct}$ & $\varphi_{\mathit{NI}}$ & $\forall\exists$ & \cmark & 10 & hopt & 1.6 & 64  & \textbf{1.3} \\
			3-Thread$_\mathit{incorrect}$& $\varphi_{\mathit{NI}}$ & $\forall\exists$ & \xmark & 57 & hpes & 12.8 & 368  & \textbf{7.7} \\
			\arrayrulecolor{black!20}\midrule
			$\mathit{NRP} : T_\mathit{correct}$  & $\varphi_{\mathit{fair}}$ & $\exists\forall$  & \cmark& 15 & hopt & 1.3 & 55  & \textbf{0.5} \\
			$\mathit{NRP} : T_\mathit{incorrect}$ & $\varphi_{\mathit{fair}}$ & $\exists\forall$ & \cmark$^!$ & 15 & hopt & 1.4 & 54  & \textbf{0.8} \\
			\arrayrulecolor{black!20}\midrule 
			$\mathit{Mutant}$  & $\varphi_{\mathit{mut}}$ & $\exists\forall$ & \cmark & 8 & hopt & 1.1 & 32  & \textbf{0.8} \\
			\arrayrulecolor{black}\bottomrule
		\end{tabular}
	\end{center}
\end{table}

In this section, we evaluate \AutoHyper{} on challenging symbolic models (NuSMV models \cite{nusmv}) that were used by Hsu et al.~\cite{HsuSB21} to evaluate \texttt{HyperQube}.

The properties we verify cover a wide range of properties. 
For example, we verify that Lamport's bakery algorithm \cite{Lamport74a} does not satisfy various symmetry properties (as the algorithm prioritizes processes with a lower ticket ID); 
We check linearizability\footnote{Linearizability asserts that any execution of a concurrent data structure corresponds to a sequential execution, which is naturally expressed as a $\forall\exists$ hyperproperty.} \cite{HerlihyW90} on the SNARK datastructure \cite{DohertyDGFLMMSS04} and identify a previously known bug;
And, we generate model-based mutation test cases using the approach proposed by Fellner et al.~\cite{FellnerBW21}.
Further details on the benchmarks are provided in \cite{HsuSB21}.

We check each instance using both \texttt{HyperQube} and \AutoHyper{} and depict the results in \Cref{tab:bmc}.\footnote{
	For the two verification instances (Bakery$_3$,$\varphi_{S3}$) and ($\mathit{NRP} : T_\mathit{incorrect}$, $\varphi_{\mathit{fair}}$) \texttt{HyperQube} provides the wrong verification result.
	We mark such instances with a ``$!$'' to avoid confusion when comparing \Cref{tab:bmc} with \cite[Table 2]{HsuSB21}.
	In particular, the supposedly unfair version of the NRP protocol is, in fact, fair. 
}
When using \AutoHyper{} we always apply \spot{}'s inclusion checker.\footnote{The automata use a symbolic alphabet with up to 18 letters.
	A conversion to an explicit alphabet -- as required for \rabit{}, \bait{}, and \forklift{} -- is thus infeasible (this would require $2^{18}$ symbols). \label{fn:symbolic_alphabet}}
For \texttt{HyperQube} we use the unrolling semantics and unrolling depth listed in \cite[Table 2]{HsuSB21}.
We observe that for most instances -- despite using explicit state methods and thus being complete (cf.~\Cref{sec:bmc_vs_explict}) -- \AutoHyper{} performs on par with \texttt{HyperQube}.
On instances using Lamport's bakery algorithm, BMC only needs to unroll to very shallow depths, resulting in very efficient solving, whereas  \AutoHyper{}'s running time is dominated by \spot{}'s LTL-to-NBA translation (consuming up to 98\% of the total time).
Conversely, on the large SNARK example, \AutoHyper{} performs significantly better. 

\subsection{Hyperproperties for Path Planning}

As a last set of benchmarks, we use planning problems for robots encoded into HyperLTL as proposed by Wang et al.~\cite{0044NP20}.
For example, the synthesis of a shortest path can be phrased as a $\exists\forall$ property that states that there exists a path to the goal such that all alternative paths to the goal take at least as long. 
Wang et al.~\cite{0044NP20} propose a solution to check the resulting HyperLTL property by encoding it in first-order logic, which is then solved by an SMT solver.
While not competitive with state-of-the-art planning tools, HyperLTL allows one to express a broad range of problems (shortest path, path robustness, etc.) in a very general way. 
Hsu et al.~\cite{HsuSB21} observe that the QBF encoding implemented in \texttt{HyperQube} outperforms the SMT-based approach by Wang et al.~\cite{0044NP20}.
In this section, we evaluate \AutoHyper{} on these planning-hyperproperties and compare it with \texttt{HyperQube}\footnote{\AutoHyper{} is intended as a model checking tool, i.e., it only checks if a property holds or is violated. However, as we show in \ifFull{Appendix \ref{app:planningViLI}}{the full version \cite{fullVersion}}, we could use the counterexamples returned by the inclusion checker to \emph{synthesize} an actual plan. }.

We depict the results in \Cref{tab:planning}. 
It is evident that \AutoHyper{} outperforms \texttt{HyperQube}, sometimes by orders of magnitude.
This is surprising as planning problems (which are essentially reachability problems) on symbolic systems should be advantageous for symbolic methods such as BMC. 
The large size of the intermediate QBF indicates that a more optimized encoding (perhaps specific to path planning) could improve the performance of BMC on such examples. 

\begin{table}[!t]
	\caption{We evaluate \texttt{HyperQube} and \AutoHyper{} on hyperproperties that encode the existence of a shortest path ($\varphi_\mathit{sp}$) and robust path ($\varphi_\mathit{rp}$).
		We give the specification (Spec), the size of the grid (Grid), and the times taken by \texttt{HyperQube} and \AutoHyper{} ($t$).
		For \texttt{HyperQube}, we additionally give the unrolling depth used ($k$) and the file size of the QBF generated (|QBF|).
		For \AutoHyper{}, we additionally give the size of the generated explicit state space (Size).
		Times are given in seconds. The timeout is set to 20 min and indicated by a ``-''.
	}\label{tab:planning}
	
	\begin{center}
		\small
		\def\arraystretch{1.2}
		\setlength\tabcolsep{3mm}
		\begin{tabular}{cc@{\hspace{10mm}}lll@{\hspace{10mm}}ll}
			\toprule
			&& \multicolumn{3}{@{}c@{\hspace{10mm}}}{\textbf{HyperQube \cite{HsuSB21}}} & \multicolumn{2}{c}{\textbf{AutoHyper}}\\
			\cmidrule[1pt](l{-1mm}r{9mm}){3-5}
			\cmidrule[1pt](l{-1mm}){6-7}
			\textbf{Spec} & \textbf{Grid}  & $\boldsymbol{k}$ & \textbf{|QBF|} & $\boldsymbol{t}$ & \textbf{Size} &  $\boldsymbol{t}$ \\
			\midrule
			\multirow{4}{*}{$\varphi_{\mathit{sp}}$} & $10 \times 10$ & 20 & 8 MB & 4.6 & 146  & \textbf{0.68} \\
			& $20 \times 20$  & 40 & 26 MB  & 168.1 & 188 & \textbf{1.50} \\
			& $40 \times 40$ & 80 & - & - & 408 &\textbf{22.70}\\
			& $60 \times 60$  & 120 & - & - & 404 & \textbf{88.83} \\
			\arrayrulecolor{black!20}\midrule 
			\multirow{3}{*}{$\varphi_{\mathit{rp}}$} & $10 \times 10$  & 20  & 13 MB  & 4.2 & 266 & \textbf{0.59} \\
			& $20 \times 20$  & 40 & 84 MB & 22.4 & 572 & \textbf{0.78}\\
			& $40 \times 40$  & 80 & 419 MB & 265 & 1212 & \textbf{1.58}\\
			& $60 \times 60$ & 120 & - & - & 1852 &\textbf{3.70}\\
			\arrayrulecolor{black}\bottomrule
		\end{tabular}
	\end{center}
\end{table}

\subsection{Bounded vs.~Explicit-State Model Checking}\label{sec:bmc_vs_explict}

Bounded model checking has seen remarkable success in the verification of trace properties and frequently scales to systems whose size is well out of scope for explicit-state methods \cite{CoptyFFGKTV01}.
Similarly, in the context of \emph{alternation-free} hyperproperties, symbolic verification tools such as \texttt{MCHyper} \cite{FinkbeinerRS15} (which internally reduces to the verification of a circuit using \texttt{ABC} \cite{BraytonM10}) can verify systems that are well beyond the reach of explicit-state methods. 
In contrast, in the context of model checking for hyperproperties that involve \emph{quantifier alternations}, our findings make a strong case for the use of explicit-state methods (as implemented in \AutoHyper{}):

First, compared to symbolic methods (such as BMC), explicit-state model checking is currently the only method that is \emph{complete}.
While BMC was able to verify or refute all properties in \Cref{tab:bmc,tab:planning}, many instances cannot be solved with the current BMC encoding. 
As a concrete example, BMC can \emph{never} verify formulas whose body contains simple invariants (such as \ref{prop:GNIintro}) and can thus not verify any of the instances in \Cref{tab:gniInstances}.
The most significant  advantage of explicit-state MC (as implemented in \AutoHyper{}) is thus that it is both push-button and complete, i.e., it can -- at least in theory -- verify or refute all properties. 

Second, the performance of \AutoHyper{} seems to be \emph{on-par} with that of BMC and frequently outperforms it (even by several orders of magnitude, cf.~\Cref{tab:planning}). 
We stress that this is despite the fact that for the evaluation of \texttt{HyperQube} we already fix an unrolling depth and unrolling semantics, thus creating favorable conditions for \texttt{HyperQube}.\footnote{
	In \Cref{tab:bmc,tab:planning}, we perform a single query with a fixed unrolling depth $k$ and semantics, i.e., we already know if we want to show satisfaction or violation and the depth needed to show this (as done in \cite{HsuSB21}).
	In a classical BMC loop, we would check for satisfaction and violation with an incrementally increasing unrolling depth and thus perform roughly $2k$ many QBF queries where $k$ is the least bound for which satisfaction or violation can be established (if this bound even exists).}
While BMC for trace properties reduces to SAT solving, BMC of hyperproperties reduces to QBF solving; a problem that is much harder and has seen less supported by industry-strength tools. 
It is, therefore, unclear whether the advance of modern QBF solvers can improve the performance of hyperproperty BMC, to the same degree that the advance of SAT solvers has stimulated the success of BMC for trace properties. 
Our findings seem to indicate that, at the moment, QBF solving (often) seems inferior to an explicit (automata-based) solving strategy.

\section{Evaluating Strategy-based Verification}\label{sec:eval3}
\vspace{-1mm}

So far, we have used \AutoHyper{} to check hyperproperties on instances arising in the literature.  
In this last section, we demonstrate that \AutoHyper{} also serves as a valuable baseline to evaluate different (possibly incomplete) verification methods. 
Here we focus on strategy-based verification (SBV), i.e., the idea of automatically synthesizing a strategy that resolves existential quantification in $\forall^*\exists^*$ HyperLTL properties \cite{CoenenFST19,BeutnerF22b}.

\subsection{Effectiveness of Strategy-based Verification}\label{sec:SBVSuccess}

\begin{wrapfigure}{R}{0.6\textwidth}
	\centering
	\includegraphics[width=\linewidth]{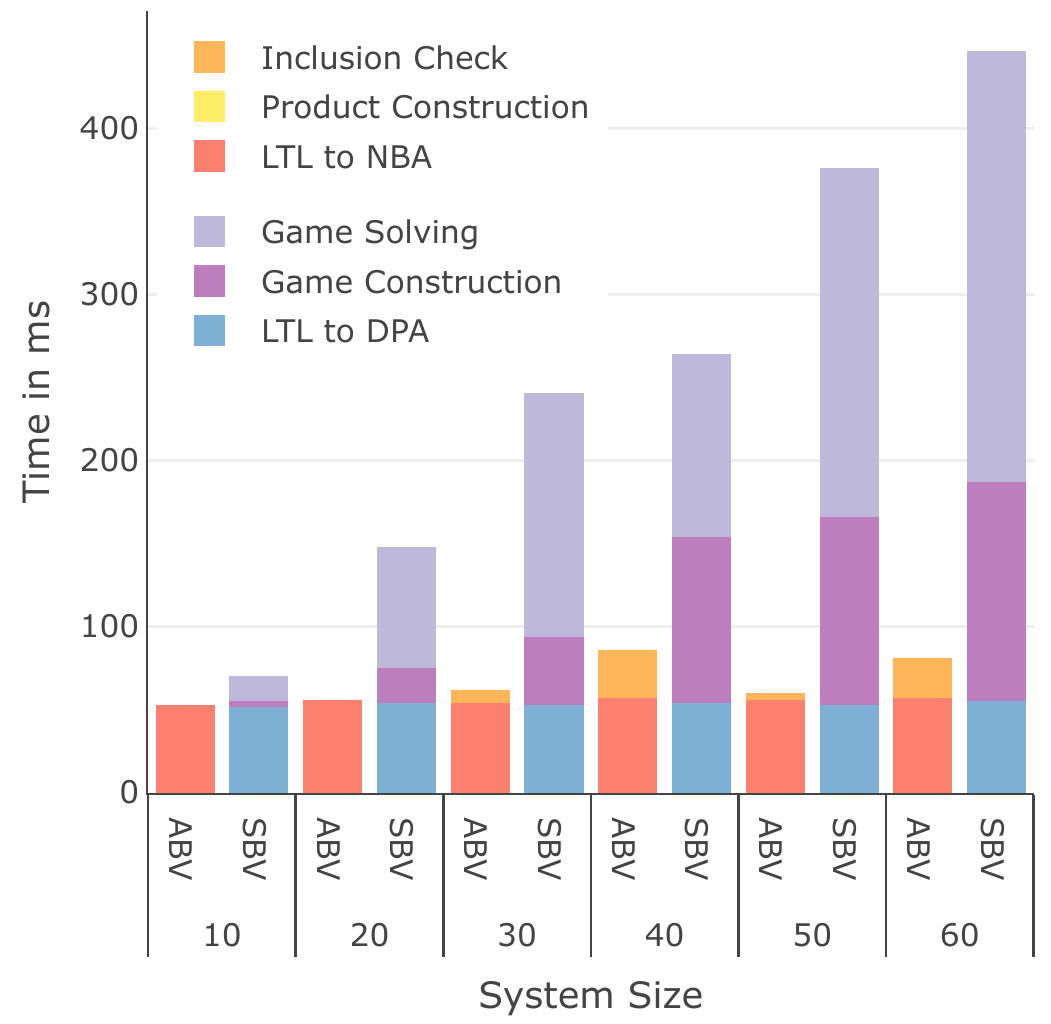}
	\caption{Comparison of ABV and SBV on instances of varying system size and fixed property size of 20.
		We generate 100 random instances for each size and take the average over the fastest $L$ instances, where $L$ is the minimum number of instances solved within a 5s timeout by both methods. }\label{fig:runtimesABVvsSBVSystemSize}
\end{wrapfigure}

SBV is known to be incomplete \cite{CoenenFST19,BeutnerF22b}.
However, due to the previous lack of \emph{complete} tools for verifying $\forall^*\exists^*$ properties, a detailed study into how effective SBV is in practice was impossible on a larger scale (i.e., beyond hand-crafted examples).
With \AutoHyper, we can, for the first time, rigorously evaluate SBV.
We use the SBV implementation from \cite{BeutnerF22b}, which synthesizes a strategy for the $\exists$-player by translating the formula to a deterministic parity automaton (DPA) \cite{Piterman07} and phrases the synthesizes as a parity game.

We have generated random transition systems and properties of varying sizes and computed a ground truth using \AutoHyper{}.
We then performed SBV (recall that SBV can never show that a property does not hold and might fail to establish that it does).
We find that for our generated instances, the property holds in \textbf{61.1\%} of the cases, and SBV can verify the property in \textbf{60.4\%} of the cases.
Successful verification with SBV is thus possible in many cases, even without the addition of expensive mechanisms such as prophecies \cite{BeutnerF22b}.
On the other hand, our results show that random generation produces instances (albeit not many) on which SBV fails (so far, examples where SBV fails required careful construction by hand).
Reverting to SBV as the default verification strategy is thus not possible, further strengthening the case for complete model checking tools (of which \AutoHyper{} is the first). 

\subsection{Efficiency of Strategy-based Verification}\label{sec:SBVtime}

After having analyzed the effectiveness of SBV (i.e., how many instances \emph{can} be verified), we turn our attention to the efficiency of SBV.
In theory, (automata-based) model checking of $\forall^*\exists^*$ HyperLTL -- as implemented in \AutoHyper{} -- is \texttt{EXPSPACE}-complete in the specification and \texttt{PSPACE}-complete in the size of the system \cite{ClarksonFKMRS14,Rabe16}.
Conversely, SBV is $2$-\texttt{EXPTIME}-complete in the size of the specification but only \texttt{PTIME} in the size of the system \cite{CoenenFST19}.
Consequently, one would expect that ABV fares better on larger specifications and SBV fares better on larger systems (the more important measure in practice).

However, in this section, we show that this does not translate into practice (at least using the current implementation of SBV \cite{BeutnerF22b}). 
We compare the running time of \AutoHyper{} (ABV) (using \spot{}'s inclusion checker) and SBV.
We break the running time into the three main steps for each method.
For ABV, this is the LTL-to-NBA translation, the construction of the product automaton, and the inclusion check. 
For SBV, it is the LTL-to-DPA translation, the construction of the game, and the game-solving.  

We depict the average cost for varying system sizes in \Cref{fig:runtimesABVvsSBVSystemSize}.
We observe that SBV performs worse than ABV and, more importantly, scales poorly in the size of the system. 
This is contrary to the theoretical analysis of ABV and SBV.
As the detailed breakdown of the running time suggests, the poor performance is due to the costly construction of the game and the time taken to solve the game. 
An almost identical picture emerges if we compare ABV in SBV relative to the property size (we give a plot in \ifFull{Appendix \ref{app:specSizeABVSBV}}{the full version \cite{fullVersion}}). 
While, in this case, the results match the theory (i.e., SBV scales worse in the size of the specification), we find that the bottleneck for SBV is not the LTL-to-DPA translation (which, in theory, is exponentially more expensive than the LTL-to-NBA translation used in ABV), but, again the construction and solving of the parity game.  

We remark that the SBV engine we used \cite{BeutnerF22b} is not optimized and always constructs the full (reachable) game graph.
The poor performance of SBV can be attributed to the fact that the size of the game does, in the worst case, scale quadratically in the size of the system (when considering $\forall^1\exists^1$ properties). 
This is amplified in dense systems (i.e., systems with many transitions), as, with increasing transition density, the size of the parity games approaches its worst-case size (see \ifFull{Appendix \ref{app:effSizeABVSBV}}{the full version \cite{fullVersion}}).
In contrast, the heavily optimized inclusion checker (in this case \spot{}) seems to be able to check inclusion in almost constant time (despite being exponential in theory).
This efficiency of mature language inclusion checkers is what enables \AutoHyper{} to achieve remarkable performance that often exceeds that of symbolic methods such as BMC (cf.~\Cref{sec:eval2}) and further strengthens the practical impact of \Cref{prop:li}.

\section{Conclusion}

In this paper, we have presented \AutoHyper, the first complete model checker for HyperLTL with an arbitrary quantifier prefix.
We have demonstrated that \AutoHyper{} can check many interesting properties involving quantifier alternations and often outperforms symbolic methods such as BMC, sometimes by orders of magnitude. 
We believe the biggest advantage of \AutoHyper{} to be its push-button functionality combined with its completeness:
As a user, one does not need to worry whether \AutoHyper{} is applicable to a particular property (different from, e.g., SBV or BMC) and does not need to provide hints (e.g., in the form of explicit strategies in SBV).

Apart from evaluating \AutoHyper's performance on a range of benchmarks, we have used \AutoHyper{} to \textbf{(1)} compare various backend language inclusion checkers, \textbf{(2)} explore practical verification beyond one quantifier alternation (which is not as infeasible as suggested by the theory),  and \textbf{(3)} rigorously evaluate the effectiveness and efficiency of strategy-based verification in practice (which, different than suggested by a theoretical analysis, performs worse than automata-based methods).

\subsubsection{Acknowledgments.}

This work was partially supported by the DFG in project 389792660 (Center for Perspicuous Systems, TRR 248), and by by the ERC Grant HYPER (No.~101055412).
R.~Beutner carried out this work as a member of the Saarbrücken Graduate School of Computer Science.

\section*{Data Availability Statement}

\AutoHyper{} and all experiments are available at \cite{artifact}. 

%
%
\bibliographystyle{splncs04}
\bibliography{references}

\iffullversion

\appendix

\newpage

\section{HyperLTL Model Checking and Language Inclusion}\label{app:langInc}

\langInc*
\begin{proof}
	For the first direction we assume that $\lang{\calA^n_\calT} \subseteq \lang{\calA_\varphi}$ and need to show that $\calT \models \dot{\varphi} = \forall \pi_1. \ldots \forall \pi_n. \varphi$.
	So let $t_1, \ldots, t_n \in \traces{\calT}$ be arbitrary. 
	By definition we have $\zip(t_1, \ldots, t_n) \in \lang{\calA^n_\calT} $ and so by assumption $\zip(t_1, \ldots, t_n) \in  \lang{\calA_\varphi}$. 
	By $\calT$-equivalence of $\calA_\varphi$ this implies $[\pi_1 \mapsto t_1, \cdots, \pi_n \mapsto t_n] \models_{\traces{\calT}} \varphi$ as required.
	
	For the reverse we assume that $\calT \models \dot{\varphi}$ and need to show that $\lang{\calA^n_\calT} \subseteq \lang{\calA_\varphi}$.
	Let $t \in \lang{\calA^n_\calT}$ be arbitrary.
	By definition of $\calA_\calT$ and $\mathit{zip}$ we have $t = \zip(t_1, \ldots, t_n)$ for some $t_1, \ldots, t_n \in \traces{\calT}$.
	As $\calT \models \dot{\varphi} = \forall \pi_1. \ldots \forall \pi_n. \varphi$ we have $[\pi_1 \mapsto t_1, \cdots, \pi_n \mapsto t_n] \models_{\traces{\calT}} \varphi$.
	By $\calT$-equivalence this implies $t = \zip(t_1, \ldots, t_n) \in  \lang{\calA_\varphi}$ as required.
	\qed
\end{proof}

\section{Overview of Verification Methods}\label{app:overviewMC}

In recent years, many different methods to verify HyperLTL properties have been developed. 
In the following, we summarize the existing approaches (possibly repeating information from the main body) and briefly discuss their relative strengths and weaknesses in Section \ref{sec:strengthsAndWeakness}.

\paragraph{Alternation-free HyperLTL and Manual Strengthening.}
HyperLTL properties without quantifier alternations can be reduced to checking a trace property on the self-composition of the system \cite{BartheDR11} and thus have the same model checking complexity as LTL (\texttt{NL}-complete in the size of the system and \texttt{PSPACE}-complete in the size of the formula).
This verification approach for HyperLTL has been implemented in \texttt{McHyper} \cite{FinkbeinerRS15}.
By using \texttt{ABC} \cite{BraytonM10} as the backend solver, \texttt{McHyper} is applicable to circuits with thousands of latches.
This is well beyond the reach of explicit-state model checking approaches.

While \texttt{McHyper} supports only alternation-free formulas, Finkbeiner et al.~\cite{FinkbeinerRS15} and D'Argenio et al.~\cite{DArgenioBBFH17} use it to verify properties involving quantifier alteration by \emph{manually} strengthening them into alternation-free formulas.
For example, a $\forall \exists$ property can be strengthened by replacing the existential quantification with a universal one, resulting in a $\forall\forall$ property.
The soundness of this strengthening must be argued manually and cannot be checked automatically. 
Moreover, such strengthening is specific to a property and -- more importantly -- also to the model; the method is thus not applicable to general HyperLTL MC. 

\paragraph{Strategy-based Verification.}

Coenen et al.~\cite{CoenenFST19} verify $\forall^*\exists^*$ properties by instating the existential quantification with an explicit strategy.
This method -- which we refer to as strategy-based verification (SBV) -- comes in two flavors: 
either the strategy is provided by the user or the strategy is synthesized. 
In the former case, model checking directly reduces to checking an alternation-free formula (which is cheap).
In fact, model checking with a given strategy is as expensive as checking alternation-free HyperLTL (\texttt{NL}-complete in the system) and can be performed on systems of considerable size. 
On the other hand, supplying an explicit strategy requires deep insight into the system. 
Different from BMC (as in \texttt{HyperQube}) or explict-state MC (as in \AutoHyper), verification with a given strategy is thus no push-button technique which is why we did not evaluate against it.
In the latter case, the method works fully automatically but requires a more expensive strategy synthesis.
SBV with automated strategy synthesis is, in theory, more efficient than ABV (in the size of the system):
In theory, (automata-based) model checking of $\forall^*\exists^*$ HyperLTL -- as implemented in \AutoHyper{} -- is \texttt{EXPSPACE}-complete in the specification and \texttt{PSPACE}-complete in the size of the system \cite{ClarksonFKMRS14,Rabe16}.
Conversely, SBV is $2$-\texttt{EXPTIME}-complete in the size of the specification but \texttt{PTIME} in the size of the system \cite{CoenenFST19}.
SBV is, however, incomplete as the strategy resolving existentially quantified traces only observes finite prefixes of the universally quantified traces (for examples see, e.g., \cite{BeutnerF22b}).
While prophecies have recently been proposed as a countermeasure to ensure completeness \cite{BeutnerF22b}, no efficient method for the synthesis of prophecies exists.
Current synthesis approaches for prophecies are limited to systems with $\leq 20$ states \cite{BeutnerF22b}. 

\newcommand*\emptycirc[1][1ex]{\tikz\draw (0,0) circle (#1);} 
\newcommand*\halfcirc[1][1ex]{%
	\begin{tikzpicture}
		\draw[fill] (0,0)-- (90:#1) arc (90:270:#1) -- cycle ;
		\draw (0,0) circle (#1);
\end{tikzpicture}}
\newcommand*\fullcirc[1][1ex]{\tikz\fill (0,0) circle (#1);} 

\begin{figure}[!t]
	\centering
		\small
	\scalebox{0.85}{
		\def\arraystretch{1.5}
		\setlength\tabcolsep{3mm}
		\begin{tabular}{@{\hspace{0mm}}l@{\hspace{2mm}}c@{\hspace{3mm}}c@{\hspace{3mm}}c@{\hspace{3mm}}c@{\hspace{3mm}}c@{\hspace{3mm}}c@{}}
			&\textbf{Sound} & \textbf{AA}& \textbf{Symbolic} & \textbf{PB} & \textbf{NL}  & \textbf{Complete}\\
			\textbf{Manual Strengthening \cite{FinkbeinerRS15,DArgenioBBFH17} }&\halfcirc&\halfcirc&\fullcirc &\emptycirc &\fullcirc&\emptycirc\\
			\textbf{SBV[Explicit Strategy] \cite{CoenenFST19}} &\fullcirc&\emptycirc&\fullcirc &\emptycirc &\fullcirc&\emptycirc\\
			\textbf{SBV[Strategy Synthesis] \cite{BeutnerF22b,BeutnerF22a}} & \fullcirc&\emptycirc&\halfcirc &\fullcirc &\emptycirc&\emptycirc\\
			\textbf{BMC \cite{HsuSB21}} &\fullcirc&\fullcirc& \fullcirc &\fullcirc &\emptycirc&\emptycirc\\
			\textbf{Explicit ABV [This work]} &\fullcirc&\fullcirc& \emptycirc &\fullcirc &\emptycirc &\fullcirc\\
		\end{tabular}
	}
	
	\caption{We list existing verification approaches for \HyperLTL{} together with their advantages and disadvantage. 
		\fullcirc{} means that the ``property'' holds, \emptycirc{{}} that it does not, and \halfcirc{} indicates that no clear answer is possible.
		We abbreviate:
		\textbf{AA} = Arbitrary alternation; \textbf{Symbolic} = Support for symbolic systems; \textbf{PB} = Push-button technique; \textbf{NL}=very efficient (\texttt{NL}) model checking in the size of the system. }\label{fig:comp}
\end{figure}

\paragraph{Bounded Model Checking.}

Hsu et al.~\cite{HsuSB21} proposed a bounded model checking (BMC) procedure for HyperLTL.
Similar to BMC for trace properties (see e.g., \cite{BiereCCZ99}) the formula is unfolded up to a fixed depth and pending obligations beyond that depth are either treated pessimistically (allowing to witness the satisfaction of a formula) or optimistically (allowing to witness the violation of a formula).
While BMC for trace properties reduces to a SAT-solving, BMC for hyperproperties naturally reduces to a QBF-solving. 
The BMC approach is implemented in \texttt{HyperQube}.
As usual for bounded methods, BMC for HyperLTL is incomplete. 
For example, it can never show that a system satisfies a hyperproperty of the form $\quant_1 \pi_1. \ldots \quant_n \pi_n. \ltlg \psi$ such as \ref{prop:GNIintro}.
In case the system terminates after a fixed number of steps (i.e., reaches a state whose only outgoing transition is to itself), pending obligations in that state can be evaluated precisely (as the entire future execution is fixed).
 Hsu et al.~\cite{HsuSB21} use this observation in their so-called \emph{halting} semantics (which comes in both an optimistic and a pessimistic flavor).
 While this improves the accuracy of BMC (for example, the halting semantics can verify GNI if all executions of a system terminates after $k$ steps and the unrolling bound is some $k' \geq k$), most reactive systems are inherently non-terminating.
 All of the examples considered in \Cref{tab:gniInstances} are non-terminating (as is usual for rwsctive programs), so the halting semantics offers no advantage -- \ref{prop:GNIintro} cannot be verified in any of those systems. 

\subsection*{Strengths and Weaknesses}\label{sec:strengthsAndWeakness}

An informal but illustrative overview of the relative strengths and weaknesses of each method is depicted in \Cref{fig:comp}.
As expected, all methods are sound (although for the manual strengthening the soundness argument cannot be checked automatically).
Different from BMC and ABV, SBV is limited to $\forall^*\exists^*$ formulas and cannot cope with arbitrary quantifier alternations. 
The biggest disadvantage of explicit-state MC is that it is not applicable to symbolic systems (such as circuits) and requires a prior conversion to an explicit system.
Nevertheless, in this paper, we show that -- at least for the existing set of benchmarks -- explicit MC performs on-par, or even outperforms, many symbolic methods by internally converting to an explicit-state representation.
SBV with strategy synthesis is, in theory, applicable to symbolic systems by solving a symbolic parity game, but this is currently not supported by any tool (hence the \halfcirc).
Manual strengthening and SBV with a given strategy are both very cheap (i.e., \texttt{NL} in the size of the system) but require significant manual effort that often requires domain knowledge. 
In contrast, SBV with strategy synthesis, BMC, and ABV and theoretically more expensive but ``push-button techniques'', i.e., they require no manual input by the user. 
Performance-wise BMC occupies a useful middle ground: while not as efficient as alternation-free methods, it reduces to the QBF problem that can, in theory, be solved faster than automaton-theoretic problems encountered in explicit-state model checking.
As we show in \Cref{sec:eval2}, this does not translate to practice and ABV can often compete with the performance of BMC.
Finally, the biggest motivation for the study of ABV (in the form of \AutoHyper) is its completeness; \AutoHyper{} is the first tool that can verify arbitrary HyperLTL properties in a sound-and-complete way.

\section{HyperLTL MC as Language Inclusion Benchmarks}\label{app:liBench}

When comparing \bait{} and \rabit{} in detail, the instances arising during HyperLTL MC seem to heavily favor \rabit{} (see \Cref{fig:largerInstances}).
This is in contrast to the observations made in \cite{DoveriGPR21}, who found that \bait{} outperforms \rabit{} on a majority of existing language inclusion benchmarks. 
Verification of $\forall^*\exists^*$ HyperLTL is, therefore, a natural candidate to extend the existing set of automaton inclusion benchmarks (an area where benchmarks are still sparse, see, e.g., \cite{DoveriGPR21}).
The resulting benchmarks seem to cover instances different from those found in existing benchmark sets (as indicated by the drastically different performance of \bait{} and \rabit{} on both benchmark families).

\section{Plan Synthesis as HyperLTL Model Checking}\label{app:planningViLI}

We note that path planning problem as studied by Wang et al.~\cite{0044NP20}, involve the \emph{synthesis} of a plan. 
If we wish to synthesize the shortest path, we can phrase the problem as a formula $\varphi = \exists \pi\ldot \forall \pi'. \psi$ wher $\psi$ states that $\pi$ encodes a path that reaches the goal and all other paths $\pi'$ take at least as long as $\pi$.
A concrete witness trace for $\pi$ is then an optimal solution to the planning problem.  
When checking $\varphi$ by reducing an SMT query (as done by Wang et al.~\cite{0044NP20}) or to a QBF problem (as done by Hsu et al.~\cite{HsuSB21}) we can easily extract a concrete solution from a satisfying model of the encoding (which is either $\exists^*\forall^*$ SMT or an $\exists^*\forall^*$ QBF formula).
Synthesizing a concrete path using the automaton-based approach implemented in \AutoHyper{} is less straightforward. 
However, we argue that we can actually use the approach to synthesize a concrete solution, provided the used language inclusion checker provides counterexamples (which is the case for most solver such as \rabit{}, \bait{}, and \forklift{}). 

In \AutoHyper{}, we check the negated $\neg \varphi = \forall \pi\ldot \exists \pi'. \neg \psi$ formula as this allows us to make use of language inclusion checkers.
As the model satisfies $\varphi$ (all models that have a shortest path will satisfy $\varphi$), it violates $\neg \varphi$, so by \Cref{prop:li}, the language inclusion that is checked does not hold. 
A trace $t$ that witnesses this non-inclusion during the model checking of $\neg \varphi$ property is thus contained in the system but has no witness trace; Or phrased equivalently, $t$ is a witness for the existential quantification in $\varphi$.
Language inclusion checkers that return counterexamples can thus be used for optimal path planning and (as shown in \Cref{tab:planning}) outperform the QBF-based approach.

\section{Details for \Cref{sec:eval3}}\label{app:effSizeABVSBV}

\begin{figure}[!t]
	\begin{subfigure}[t]{0.48\linewidth}
		\centering
		\includegraphics[width=\linewidth]{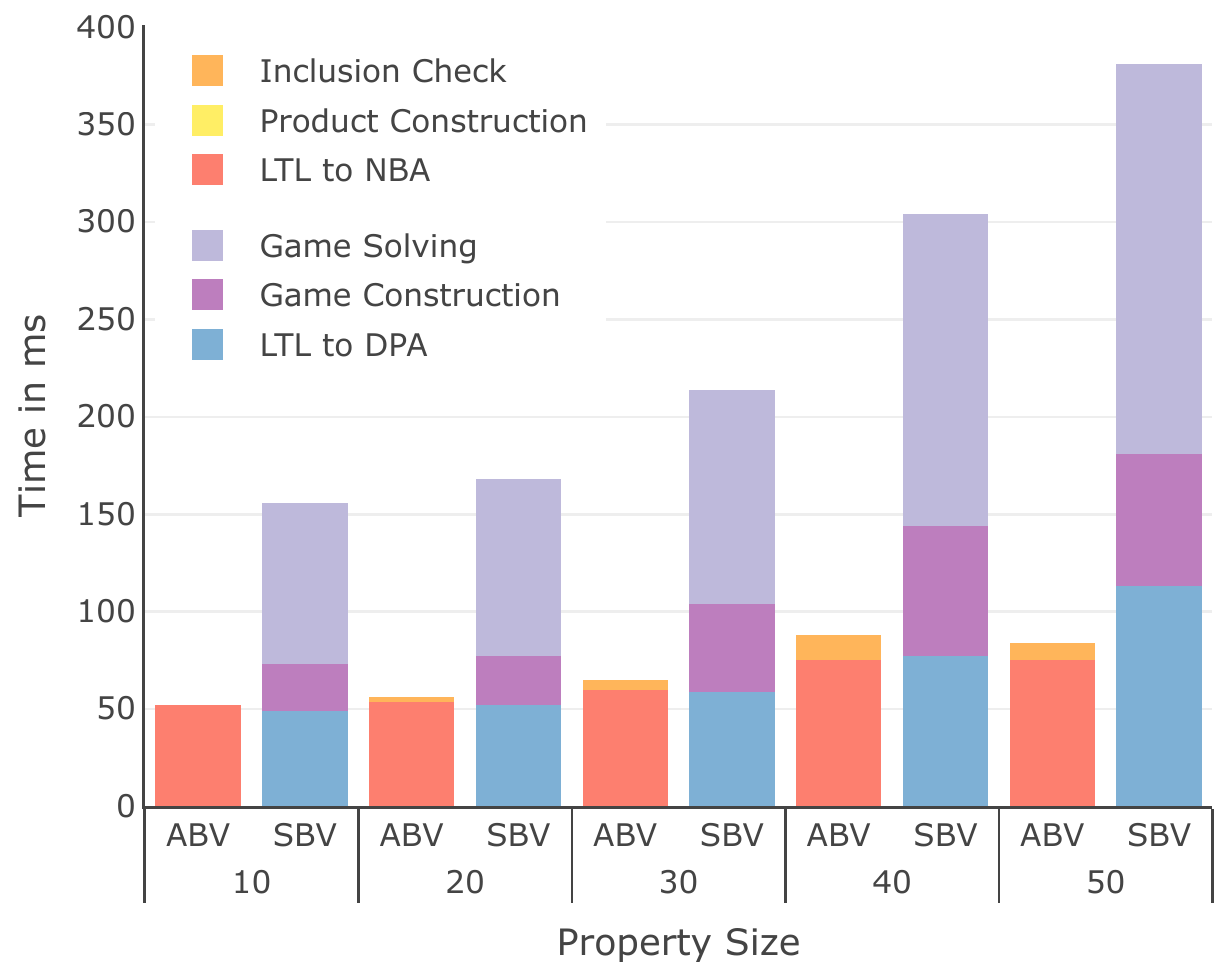}
		\vspace{-3mm}
		\subcaption{Running times of ABV and SBV with varying property size.  We fix the system size 30. }\label{fig:runtimesABVvsSBVPropSize}
	\end{subfigure}\hfill
	\begin{subfigure}[t]{0.48\linewidth}
		\centering
		\includegraphics[width=\linewidth]{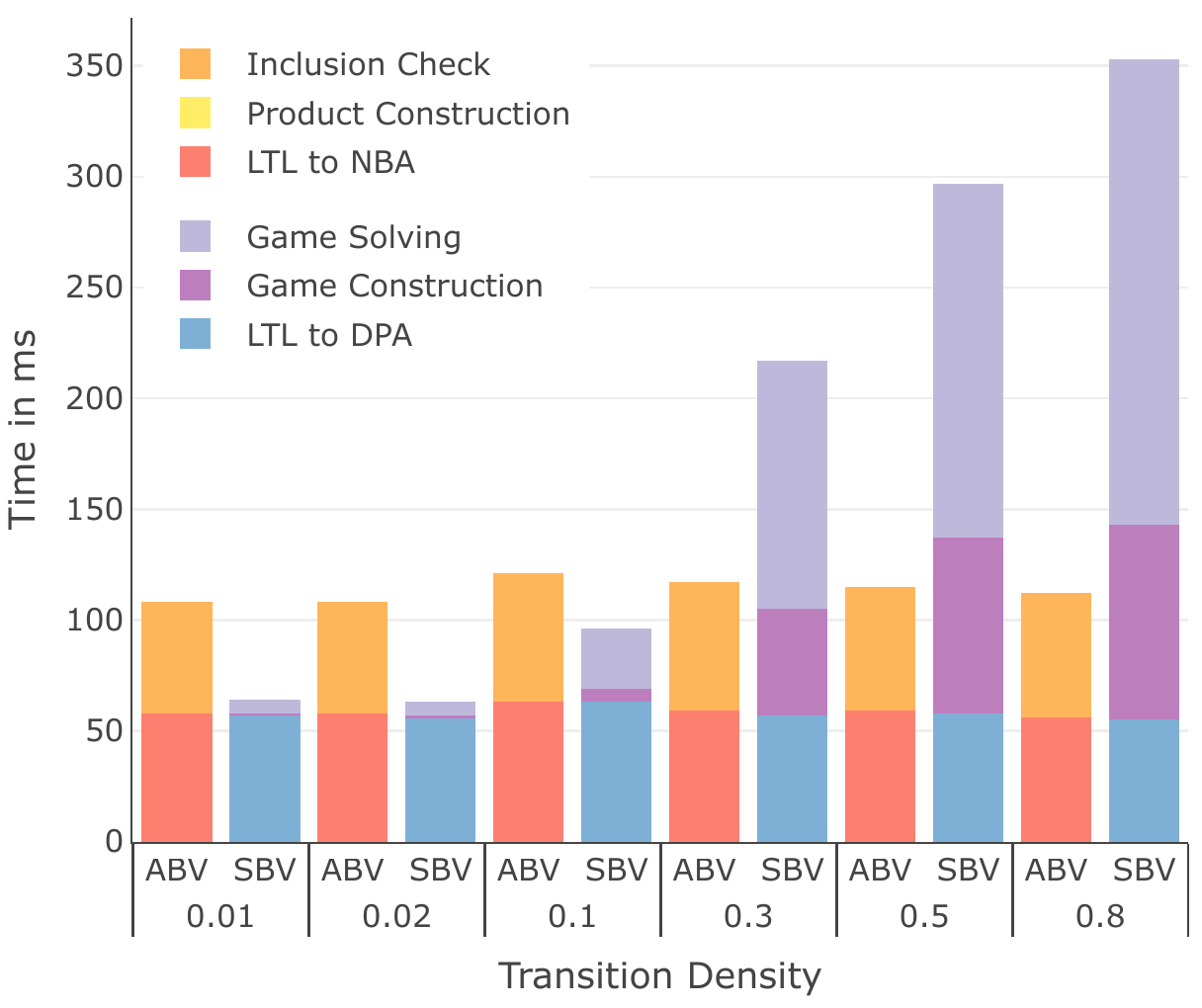}
		\vspace{-3mm}
		\subcaption{Running times of ABV and SBV with varying transition density. We fix the system size to 30, and the property size to 20.}\label{fig:runtimesABVvsSBVDensity}
	\end{subfigure}
	
	\caption{We compare ABV and SBV on instances with varying property size (in \Cref{fig:runtimesABVvsSBVPropSize}) and transition density (in \Cref{fig:runtimesABVvsSBVDensity}). For each size we generate 100 random $\forall^1\exists^1$ properties. 
		The timeout is set to 5s. 
		We take the average over the fastest $L$ instances, where $L$ is the minimum of the instances solved within a 5s timeout by both ABV and SBV.}
\end{figure}

\subsection{Efficiency of SBV -- Impact of Specification Size}\label{app:specSizeABVSBV}

Similar to \Cref{fig:runtimesABVvsSBVSystemSize},  \Cref{fig:runtimesABVvsSBVPropSize} compares the running time of ABV and SBV in terms of the property size. 
While these results match the theory, i.e., SBV scales worse than ABV in the size of the specification, the detailed breakdown of the total running time shows that the decrease in performance is not for the reason suggested by theory. 
While the limiting factor should be a much more expensive LTL-to-DPA translation (which is exponentially more expensive than a LTL-to-NBA translation involved in ABV) the actual blowup stems from the construction and solving time of the parity game.

\subsection{Efficiency of SBV -- Impact of Transition Density}\label{app:specSizeABVDensity}

We perform a similar experiment as in \Cref{fig:runtimesABVvsSBVSystemSize,fig:runtimesABVvsSBVPropSize} but do no vary the size of the transition system but only its density. 
\Cref{fig:runtimesABVvsSBVDensity} depicts the results.
We observe that with increasing density, SBV scales very poorly.
This confirms that the poor performance of SBV is due to the size of the game space, which, with increasing transition density, approaches quadratic size (in the size of the system). 
To see this, take a system with $n$ states where each state has a unique successor. The resulting parity game constructs the $2$-fold self-composition of the system, which has $\Theta(n)$ states. At the other extreme, if each state has a transition to all other states, the resulting game has size $\Theta(n^2)$ as all pairs of states occur in the product. 

\fi

\end{document}